\documentclass[letterpaper, 10 pt, journal, twoside]{IEEEtran}
\IEEEoverridecommandlockouts
\usepackage{cite}
\usepackage{amsmath,amssymb,amsfonts}
\usepackage{bm}
\usepackage{algorithmic}
\usepackage{graphicx}
\usepackage{textcomp}
\usepackage{xcolor}
\usepackage{amsthm}
\usepackage[top=0.833in, left=0.667in, right=0.667in, bottom=0.597in, includefoot]{geometry}

\usepackage{pifont}
\newcommand{\cmark}{\textcolor{green}{\ding{51}}}%
\newcommand{\xmark}{\textcolor{red}{\ding{55}}}%
    
\makeatletter
\let\MYcaption\@makecaption
\makeatother 

\usepackage[font=footnotesize]{subcaption}

\makeatletter
\let\@makecaption\MYcaption
\makeatother

\newtheorem{definition}{Definition}
\newtheorem{lemma}{Lemma}
\newtheorem{theorem}{Theorem}
\newtheorem{remark}{Remark}

\DeclareMathOperator{\until}{\mathbf{U}}
\DeclareMathOperator{\release}{\mathbf{R}}
\DeclareMathOperator{\always}{\mathbf{G}}
\DeclareMathOperator{\eventually}{\mathbf{F}}

\begin{document}

\title{A Smooth Robustness Measure of Signal Temporal Logic for Symbolic Control
}

\author{Yann Gilpin, Vince Kurtz, and Hai Lin
        \thanks{The partial support of the National Science Foundation (Grant No. CNS-1830335, IIS-1724070) and of the Army Research Laboratory (Grant No. W911NF- 17-1-0072) is gratefully acknowledged.}
\thanks{All authors are with the Department of Electrical Engineering, University of Notre Dame, Notre Dame, IN 46556, USA, {\tt\small ygilpin@nd.edu; vkurtz@nd.edu; hlin1@nd.edu}.}}

\maketitle
\thispagestyle{empty}

\begin{abstract}
Recent years have seen an increasing use of Signal Temporal Logic (STL) as a formal specification language for symbolic control, due to its expressiveness and closeness to natural language. Furthermore, STL specifications can be encoded as cost functions using STL's robust semantics, transforming the synthesis problem into an optimization problem. Unfortunately, these cost functions are non-smooth and non-convex, and exact solutions using mixed-integer programming do not scale well. Recent work has focused on using smooth approximations of robustness, which enable faster gradient-based methods to find local maxima, at the expense of soundness and/or completeness. We propose a novel robustness approximation that is smooth everywhere, sound, and asymptotically complete. Our approach combines the benefits of existing approximations, while enabling an explicit tradeoff between conservativeness and completeness. 
\end{abstract}

\begin{IEEEkeywords}
    Autonomous systems; Robotics; Intelligent systems
\end{IEEEkeywords}

\section{Introduction}

\IEEEPARstart{T}{emporal} logics provide an intuitive way to specify high-level objectives for autonomous systems. Logics like Linear Temporal Logic (LTL), Computation Tree Logic (CTL) and Signal Temporal Logic (STL) can express many desirable system behaviors. Certain temporal logic formulas can even be derived from natural language commands \cite{dzifcak2009and,kress2007structured,lignos2015provably}. 

Given a temporal logic specification, the symbolic control problem is to design a trajectory that satisfies the specification. Early work on this problem focused on finding discrete abstractions of continuous systems. Given a finite transition system and a specification, automata-theoretic methods can be used to check or enforce satisfaction (see e.g., \cite{belta2007symbolic,kloetzer2008fully,gol2013language,plaku2016motion,belta2017formal} and references therein). However, these methods tend to scale poorly with system dimension and specification complexity.

With this in mind, much recent work has focused on STL specifications, which are defined over continuous signals \cite{maler2004monitoring}. Output trajectories provide a convenient such signal, especially for robotics \cite{raman2014model,sadraddini2015robust,pant2017smooth,mehdipour2019arithmetic}. Furthermore, STL robust semantics allows us to cast the synthesis problem as an optimization problem. Robust semantics define a ``robustness measure'' which maps a signal to a scalar which is positive only if the signal satisfies the specification. Symbolic control is as simple as finding a control sequence that maximizes robustness. 

Unfortunately, the standard robustness measure is non-smooth and non-convex for most interesting specifications, due to the use of $\min$ and $\max$ operators. Nonetheless, the optimization can be solved exactly using Mixed-Integer Programming (MIP) \cite{belta2019formal}. This approach is both sound (any satisfying trajectory found by MIP does satisfy the specification) and complete (if a satisfying trajectory exists, the MIP-based algorithm will find it). However, the MIP approach scales poorly with the system dimension, the length of trajectory considered, and the complexity of the specification. This stems from the fact that a new integer variable is introduced for every predicate at each timestep, and the fact that MIP complexity is exponential in the number of integer variables \cite{belta2019formal}.


To address this issue, recent work has focused on smooth approximations of the robustness measure. By enabling the use of fast gradient-based optimization methods, this approach tends to perform well on a wide variety of problems, and offers considerable speed and scalability improvements over MIP \cite{pant2017smooth,pant2018fly,mehdipour2019arithmetic,haghighi2019control,lindemann2019robust}, though it may not find the global maximum unless the specification is convex \cite{lindemann2019robust}.

In \cite{pant2017smooth}, the widely known Log-Sum-Exponential (LSE) approximation of $\min$ and $\max$ was used to define an approximate robustness measure. This LSE approximation is smooth everywhere, and an analytical form of the gradient accelerates optimization. Furthermore, this method is \textit{asymptotically complete}, meaning that the approximation can be arbitrarily close to the true robustness. This approach is not sound, however, since it is based on an overapproximation of the $\max$ operator. 

To address this soundness issue, \cite{mehdipour2019arithmetic} introduced an approximation of the robustness measure based on arithmetic and geometric means. This Arithmetic-Geometric Mean (AGM) robustness retains many of the computational advantages of the LSE approach \cite{mehdipour2019arithmetic}, but ensures soundness. Additionally, AGM robustness is more conservative in the sense that trajectories derived using the AGM approach tend to be more robust to external disturbances. This comes at a price, however, since the AGM robustness is not smooth everywhere \cite{mehdipour2019arithmetic}. 

In this letter, we propose a simple approximation of the robustness measure that combines the advantages of the LSE and AGM robustness approximations. Like the LSE robustness, our proposed approximation is asymptotically complete and smooth everywhere. Like the AGM robustness, our proposed approximation is an under-approximation of the true robustness measure, and therefore our method is sound. Furthermore, like the AGM robustness, our method can introduce extra conservativeness to increase robustness to external disturbances. Unlike the existing approaches, however, our approximation allows us to explicitly regulate the degree of conservativeness with several tunable parameters, and to trade off this extra conservativeness with asymptotic completeness.  

The remainder of this letter is organized as follows: a brief introduction to STL is provided in Section \ref{sec:background}. Our main results are presented in Section \ref{sec:smooth_robustness}. Section \ref{sec:simulation} provides simulation comparisons of our approach with other smooth robustness measures, and Section \ref{sec:conclusion} concludes the letter.

\section{Background}\label{sec:background}

In this letter, we consider nonlinear discrete-time systems of the following form:
\begin{equation}\label{eq:system}
\left\{
\begin{gathered}
    \bm{x}_{t+1} = f(\bm{x}_t,\bm{u}_t), \\
    \bm{y}_t = g(\bm{x}_t,\bm{u}_t),
\end{gathered}\right.
\end{equation}
where $\bm{x}_t \in { \mathcal{X} \subseteq~} \mathbb{R}^n$ is the system state, $\bm{u}_t \in { \mathcal{U} \subseteq~} \mathbb{R}^m$ is a control input, and $\bm{y}_t \in {\mathcal{Y} \subseteq~} \mathbb{R}^p$ is an output signal. Given such a system, an initial condition $\bm{x}_0$, and a bounded-time STL specification $\varphi$, our goal is to design a control sequence $\bm{u} = \bm{u}_0\bm{u}_1...\bm{u}_T$ such that the resulting output signal $\bm{y} = \bm{y}_0\bm{y}_1...\bm{y}_T$ satisfies $\varphi$. {We assume that $f(\cdot,\cdot)$ and $g(\cdot)$ are continuous and locally Lipschitz.}


\subsection{STL Robust Semantics}

In this section we give a brief introduction to the syntax and semantics of STL. Further details can be found in \cite{belta2019formal}. STL syntax is defined as
\begin{equation}\label{eq:syntax}
    \varphi := \pi \mid \lnot \varphi \mid \varphi_1 \land \varphi_2 \mid \varphi_1 \until_{[t_1,t_2]} \varphi_2,
\end{equation}
where $\pi = (\mu^\pi(\bm{y}_t)-c \geq 0)$ is a predicate defined by the function $\mu^{\pi} : \mathbb{R}^p \to \mathbb{R}$. Many STL synthesis methods require $\mu^{\pi}(\cdot)$ and system (\ref{eq:system}) to be linear \cite{belta2019formal}, but this restriction is not necessary for our approach. Negation ($\lnot$) and conjuction ($\land$) can be used to derive other boolean operations like disjuction ($\lor$) and implication ($\implies$). Similarly the temporal operator until ($\varphi_1\until_{[t_1,t_2]}\varphi_2)$ can be used to construct eventually ($\eventually_{[t_1,t_2]}\varphi$) and always ($\always_{[t_1,t_2]}\varphi$) operators. We restrict ourselves to bounded-time STL formulas, i.e., $t_2$ is finite. 

We define the meaning, of STL using qualitative or ``robust'' semantics. These semantics associate a scalar value (the robustness measure) with a signal. The signal satisfies a given specification if and only if the associated robustness measure is positive. Given signal $\bm{y} = \bm{y}_0\bm{y}_1 ... \bm{y}_T$, we denote the suffix starting at timestep $t$ as $(\bm{y},t) = \bm{y}_t\bm{y}_{t+1} ... \bm{y}_T$. The STL robust semantics are defined as follows:

\begin{definition}[STL Robust Semantics]~\label{def:robust_semantics}
\begin{itemize}
    \item $\bm{y} \vDash \varphi \iff \rho^\varphi((\bm{y},0)) > 0$
    \item $\rho^\pi( (\bm{y},t) ) = \mu^\pi(\bm{y}_t) - c$
    \item $\rho^{\lnot \varphi}( (\bm{y},t) ) = - \rho^\varphi( (\bm{y},t) )$
    \item $\rho^{\varphi_1 \land \varphi_2}( (\bm{y},t) ) = \min\big(\rho^{\varphi_1}( (\bm{y},t) ), \rho^{\varphi_2}( (\bm{y},t) ) \big)$
    \item $\rho^{\eventually_{[t_1,t_2]} \varphi}( (\bm{y},t) ) = \max_{t'\in[t+t_1, t+t_2]}\big(\rho^\varphi( (\bm{y},t') )\big)$
    \item $\rho^{\always_{[t_1,t_2]} \varphi}( (\bm{y},t) ) = \min_{t'\in[t+t_1, t+t_2]}\big(\rho^\varphi( (\bm{y},t') )\big)$
    \item $\rho^{\varphi_1 \until_{[t_1,t_2]} \varphi_2}( (\bm{y},t) ) =  \max_{t'\in[t+t_1, t+t_2]}\bigg( \\ \min\Big(\Big[\rho^{\varphi_1}( (\bm{y},t') ), \min_{t'' \in[t+t1,t']}\big(\rho^{\varphi_2}( (\bm{y},t'') )\big)\Big]^T\Big)\bigg)$.
\end{itemize}
\end{definition}

%

\subsection{Optimization and Smooth Approximations}

The STL robust semantics allow us to cast the synthesis problem (find a control sequence to satisfy a specification) as an optimization problem over the control sequence $\bm{u} = \bm{u}_0\bm{u}_1...\bm{u}_T$ as follows:
\begin{align}\label{eq:nonconvex_opt}
    \max_{\bm{u}} &~ \rho^{\varphi}((\bm{y},0)) \\    
    \text{s.t. } & \bm{x}_{t+1} = f(\bm{x}_t,\bm{u}_t) \\
                 & \bm{y}_t = g(\bm{x}_t,\bm{u}_t) \\
                 & {\bm{x}_t \in \mathcal{X}, \bm{u}_t \in \mathcal{U}, \bm{y}_t \in \mathcal{Y}}
\end{align}
If the optimal control sequence generates an output signal $\bm{y}^*$ such that $\rho^{\varphi}((\bm{y}^*,0){)} > 0$, then we have found a trajectory that ensures satisfaction of the specification. 

Unfortunately, this problem is non-smooth and non-convex in general. For the special case of linear systems and linear predicates, (\ref{eq:nonconvex_opt}) can be encoded as a Mixed Integer Program (MIP) \cite{belta2019formal}. But the MIP approach has serious scalability restrictions: the resulting algorithm has exponential complexity in both the number of predicates and the number of timesteps $T$ considered. This makes the MIP approach impractical for high-dimensional systems, complex specifications, and long-duration tasks. 

The lack of smoothness in the robust semantics is due to the use of $\min$ and $\max$ operators. To deal with this , \cite{pant2017smooth} proposed using the following smooth approximations to define an approximate robustness measure:
\begin{align}\label{eq:lse_approximation}
    & \max([a_{1},...,a_{m}]^T) \approx \frac{1}{k}\log \Big(\sum_{i=1}^{m} e^{k a_{i}} \Big) \\
    & \min([a_{1},...,a_{m}]^T) \approx -\frac{1}{k}\log \Big(\sum_{i=1}^{m} e^{-k a_{i}} \Big).
\end{align}

This log-sum-exponential (LSE) approximation is smooth, and an analytical form of the gradient exists. The resulting robustness approximation, which we will refer to as ``LSE robustness'' is also smooth everywhere, and approaches the true robustness as $k \to \infty$ \cite{pant2017smooth}. We refer to this property as \textit{asymptotic completeness}: as $k \to \infty$, positive LSE robustness becomes a necessary and sufficient condition for satisfying the specification\footnote{Note that the notion of asymptotic completeness used in this letter is a property of the robustness measure, not of the synthesis algorithm. This is in contrast to the common usage of the term ``completeness'' in the formal methods literature to refer to an algorithm which finds a solution if one exists.}. The resulting optimization is still non-convex, but smooth optimization techniques like Sequential Quadratic Programming (SQP) can find a local maximum\cite{pant2017smooth,pant2018fly}. 

These gradient-based techniques significantly outperform MIP in terms of scalability, and despite finding only local maxima, perform well on a variety of difficult problems \cite{pant2017smooth,pant2018fly,mehdipour2019arithmetic}. For finite $k$, however, LSE robustness can over-approximate the true robustness. This is because
\begin{equation*}
    \max([a_{1},...,a_{m}]^T) \leq \frac{1}{k}\log \Big(\sum_{i=1}^{m} e^{k  a_{i}} \Big),
\end{equation*}
and means that LSE robustness is not sound: a signal may have positive LSE robustness but not satisfy the specification. 

To address this issue, \cite{mehdipour2019arithmetic} proposed a new under-approximation of the robustness. This robustness measure is based on arithmetic and geometric means, and we will refer to it as ``AGM robustness''. AGM robustness is sound: positive AGM robustness is a sufficient condition for satisfying a specification\footnote{{
In this letter, we use ``soundness'' to describe when a positive robustness value implies satisfaction, and ``completeness'' to refer to the converse. This matches the standard use of the terms in reference to algorithms \cite{baier2008principles}, but is in contrast with the terminology used in \cite{mehdipour2019arithmetic}, where ``soundness'' is used to refer to both of these properties. 
}}. Additionally, the conservativeness introduced by this under-approximation leads to additional robustness against disturbances and modeling errors. The price of this additional conservativeness is that the gradient of AGM robustness is not defined everywhere. This means that convergence to a local maximum cannot be guaranteed with gradient-based optimization techniques.

In this letter, we propose a deceptively simple smooth approximation of STL robustness which combines the best of AGM and LSE robustness. Like LSE robustness, our proposed measure is asymptotically complete and smooth everywhere. Like AGM robustness, our proposed robustness measure is sound and can introduce additional conservativeness (see Table \ref{tab:comparison_table} for a visual comparison). Furthermore, our approach offers an \textit{explicit tradeoff between conservativeness and completeness}. For large parameter values, our proposed robustness measure approaches the true robustness. For smaller parameter values, our proposed robustness measure is a significant under-approximation of the true robustness and results in more conservative system behavior. This smooth robustness measure is introduced in the following section.

\section{Proposed Smooth Robustness}\label{sec:smooth_robustness}

Following the spirit of \cite{pant2017smooth}, we define a smooth robustness measure via a pair of approximations, one for maximum and one for minimum. By replacing instances of $\min$ and $\max$ in Definition \ref{def:robust_semantics}, we obtain a smooth approximation of the robustness function. 

\subsection{Several Smooth Approximations}
For minimum, we use the log-sum-exponential approximation
\begin{equation}\label{eq:ourmin}
    \widetilde{\min}([a_i,..,a_m]^T) = -\frac{1}{k_1}\log \Big(\sum_{i=1}^{m} e^{-k_1 a_{i}} \Big),
\end{equation}
where $k_1 > 0$ is an adjustable parameter. This function is an under-approximation of the true minimum, and approaches the true minimum as $k_1 \to \infty$. This is expressed formally in the following Lemma: 
\begin{lemma}\label{lemma:min}
    Consider a vector $\mathbf{a} = [a_1, a_2, \dots, a_m]^T$. The smooth minimum (\ref{eq:ourmin}) is an under-approximation of the true minimum, and an associated error bound is given by: 
     \begin{equation}\label{eq:ourminbnd}
        \min(\mathbf{a}) - \widetilde{\min}(\mathbf{a}) \leq \frac{\log(m)}{k_1}.
     \end{equation}
\end{lemma}
\begin{proof}
    The approximation is the least accurate when the input arguments are all equal: 
    \begin{equation*}
        a_1 = a_2 = \dots = a_m.
    \end{equation*}
    \begin{align*}
        \min(\mathbf{a}) - \widetilde{\min}(\mathbf{a})  &= a_m -\frac{-1}{k_1} \log{\sum_{i=1}^{m}e^{-k_1 a_i}}\\
        &= a_m -\frac{-1}{k_1} \log{(m e^{-k_1 a_m})} \leq \frac{\log{m}}{k_1},
    \end{align*}
    Where equality holds only in the worst case. Noting that $m>0$ and $k_1>0$, it is clear that $\min(\mathbf{a}) > \widetilde{\min}(\mathbf{a})$.
\end{proof}

Additionally this function has a well defined gradient: 
\begin{equation*}
    \nabla_{a_i} \widetilde{\min}(\mathbf{a}) = \frac{e^{k_1 a_i}}{\sum_{j=1}^m e^{k_1 a_j}},
\end{equation*}
which will ultimately enable more efficient optimization over our approximate robustness measure. 

\begin{table}[]
    \centering
    \caption{Comparison with related work.}
    \begin{tabular}{c|c c c }
        ~ & \cite{pant2017smooth} & \cite{mehdipour2019arithmetic} & Ours \\
        \hline
        Sound & \xmark & \cmark & \cmark \\
        (Asymptotically) Complete & \cmark & \cmark & \cmark \\
        Differentiable Everywhere &  \cmark & \xmark & \cmark
    \end{tabular}
    \label{tab:comparison_table}
\end{table}

For the maximum function, we adopt the following well-known approximation \cite{lange2014applications}: 
 \begin{equation}\label{eq:ourmax}
    \widetilde{\max}([a_{1},a_{2},\dots,a_{m}]^T) = \frac{\sum_{i=1}^{m} a_{i} e^{k_2 a_{i}}}{\sum_{i=1}^{m} e^{k_2 a_{i}}}
\end{equation}
where $k_2 > 0$ is an adjustable parameter. As with the minimum approximation (\ref{eq:ourmin}), this is an under-approximation of the true maximum, which approaches the true maximum for large $k_2$.
This is formalized in the following Lemma:

\begin{lemma}\label{lemma:max}
    Consider a vector $\mathbf{a} = [a_1, a_2, \dots, a_m]^T$. Assume (without loss of generality) that $m \geq 2$ and that $\mathbf{a}$ is sorted from largest to smallest. Then the smooth maximum (\ref{eq:ourmax}) is an under-approximation of the true maximum, and an associated error bound is given by
    \begin{equation}\label{eq:ourmaxbnd}
        \max(\mathbf{a}) - \widetilde{\max}(\mathbf{a}) \leq \frac{a_1-a_m}{\frac{e^{k_2(a_1-a_2)}}{m-1}+1}.
    \end{equation}
\end{lemma}
\begin{proof}
    \begin{align*}
        \max(\mathbf{a})-\widetilde{\max}(\mathbf{a}) &= a_1 - \frac{\sum_{i=1}^m a_i e^{k_2 a_i}}{\sum_{i=1}^m e^{k_2 a_i}} \\
        &= \frac{a_1 \sum_{i=2}^m e^{k_2 a_i}-\sum_{i=2}^m a_i e^{k_2 a_i}}{e^{k_2 a_1} + \sum_{i=2}^m e^{k_2 a_i}}.
    \end{align*}
    
    Letting $K = \sum_{i=2}^m e^{k_2 a_i}$, we have
    \begin{align*}
        \max(\mathbf{a})-\widetilde{\max}(\mathbf{a}) &\leq \frac{a_1 K-a_m K}{e^{k_2 a_1} + K}\\
        &\leq \frac{a_1-a_m}{\frac{e^{k_2 (a_1-a_2)}}{m-1}+1},
    \end{align*}
    where the last inequality holds since $K \leq (m-1)e^{k_2 a_2}$.
\end{proof}

By using an under-approximation for maximum, we can be confident that if our smooth robustness indicates satisfaction (our robustness measure is positive) then the associated trajectory does in fact satisfy the formula. Additionally this approximation has a well defined gradient: 
\begin{equation*}
    \nabla_{a_{i}} \widetilde{\max}(\mathbf{a}) = \frac{e^{k_2 a_{i}}}{\sum_{j=1}^{N} e^{k_2 a_{j}}}[1 + k_2(a_i - \widetilde{\max}(\mathbf{a}))].
\end{equation*}

\begin{remark}
    The smooth maximum (\ref{eq:ourmax}) suggests that there may be a deeper connection between our proposed smooth robustness and AGM robustness. Specifically, taking $k_2=0$, $\widetilde{\max}(\mathbf{a})$ is the arithmetic mean of vector $\mathbf{a}$.
\end{remark}

\subsection{Smooth Robustness}
We can now use the approximations $\widetilde{\min}$ and $\widetilde{\max}$ to define a new smooth robustness measure for STL:
\begin{definition}[Smooth STL Robust Semantics]~\label{def:smooth_semantics}
\begin{itemize}
    \item $\tilde{\rho}^{\varphi}((\bm{y},0)) > 0 \implies \bm{y} \vDash \varphi$
    \item $\tilde{\rho}^\pi( (\bm{y},t) ) = \mu^\pi(\bm{y}(t)) - c$
    \item $\tilde{\rho}^{\lnot \pi}( (\bm{y},t) ) = - \tilde{\rho}^\pi( (\bm{y},t) )$
    \item $\tilde{\rho}^{\varphi_1 \lor \varphi_2}( (\bm{y},t) ) = \widetilde{\max}\big(\tilde{\rho}^{\varphi_1}( (\bm{y},t) ), \tilde{\rho}^{\varphi_2}( (\bm{y},t) ) \big)$
    \item $\tilde{\rho}^{\varphi_1 \land \varphi_2}( (\bm{y},t) ) = \widetilde{\min}\big(\tilde{\rho}^{\varphi_1}( (\bm{y},t) ), \tilde{\rho}^{\varphi_2}( (\bm{y},t) ) \big)$
    \item $\tilde{\rho}^{\eventually_{[t_1,t_2]} \varphi}( (\bm{y},t) ) = \widetilde{\max}_{t'\in[t+t_1, t+t_2]}\big(\tilde{\rho}^\varphi( (\bm{y},t') )\big)$
    \item $\tilde{\rho}^{\always_{[t_1,t_2]} \varphi}( (\bm{y},t) ) = \widetilde{\min}_{t'\in[t+t_1, t+t_2]}\big(\tilde{\rho}^\varphi( (\bm{y},t') )\big)$
    \item $\tilde{\rho}^{\varphi_1 \until_{[t_1,t_2]} \varphi_2}( (\bm{y},t) ) =  \widetilde{\max}_{t'\in[t+t_1, t+t_2]}\Bigg(\\ \widetilde{\min}\Big(\Big[\tilde{\rho}^{\varphi_1}( (\bm{y},t') ), \widetilde{\min}_{t'' \in[t+t1,t']}\big(\tilde{\rho}^{\varphi_2}( (\bm{y},t'') )\big)\Big]^T\Big)\Bigg)$.
    \item {$\tilde{\rho}^{\varphi_1 \release_{[t_1,t_2]} \varphi_2}( (\bm{y},t) ) =  \widetilde{\max}_{t'\in[t+t_1, t+t_2]}\Bigg(\\ \widetilde{\min}\Big(\Big[\tilde{\rho}^{\varphi_2}( (\bm{y},t') ), \widetilde{\min}_{t'' \in[t+t1,t']}\big(\tilde{\rho}^{\varphi_1}( (\bm{y},t'') )\big)\Big]^T\Big)\Bigg)$.}
\end{itemize}
\end{definition}

\begin{remark}
    In formulating this smooth robustness, we assume (without loss of generality) that STL formulas are written in disjunctive normal form, i.e., the negation operator is only applied to predicates $\pi${, and we introduce the release operator $\varphi_1 \release_{[t_1,t_2]} \varphi_2$}. Any STL formula can be re-written to be in disjunctive normal form, and the resulting formula has length linear in the length of the original formula \cite{raman2014model}.
\end{remark}

Note that this smooth robust semantics provides only a sufficient condition for satisfaction in the general case. A positive smooth robustness measure is a necessary condition for satisfaction only in the limit as $k_1 \to \infty$, $k_2 \to \infty$.

The soundness of our proposed smooth robustness measure is stated formally in the following Theorem:
\begin{theorem}\label{thm:soundness}
    For any discrete-time signal $\bm{y}$ and any bounded-time STL specification $\varphi$, if $\tilde{\rho}((\bm{y},0)) > 0$ then the signal $\bm{y}$ satisfies the specification, i.e., $\bm{y} \vDash \varphi$.
\end{theorem}
\begin{proof}
    Recall from Lemmas \ref{lemma:min} and \ref{lemma:max} that $\widetilde{\min}(\mathbf{a}) \leq \min(\mathbf{a})$ and $\widetilde{\max}(\mathbf{a}) \leq \max(\mathbf{a})$. 
    
    It is then clear from Definitions \ref{def:robust_semantics} and \ref{def:smooth_semantics} that $\tilde{\rho}^{\varphi}((\bm{y},t)) < \rho^\varphi((\bm{y},t))$ for any $\varphi$ in disjunctive normal form:
    \begin{itemize}
        \item $\tilde{\rho}^\pi((\bm{y},t)) = \rho^\pi((\bm{y},t))$
        \item $\tilde{\rho}^{\lnot \pi}((\bm{y},t)) = \rho^{\lnot \pi}((\bm{y},t))$
        \item $\tilde{\rho}^{\varphi_1 \land \varphi_2}((\bm{y},t)) \leq \rho^{\varphi_1 \land \varphi_2}((\bm{y},t))$
        \item $\tilde{\rho}^{\varphi_1 \lor \varphi_2}((\bm{y},t)) \leq \rho^{\varphi_1 \lor \varphi_2}((\bm{y},t))$
        \item $\tilde{\rho}^{\varphi_1 \until_{[t_1,t_2]} \varphi_2}((\bm{y},t)) \leq \rho^{\varphi_1  \until_{[t_1,t_2]} \varphi_2}((\bm{y},t))$.
        \item {$\tilde{\rho}^{\varphi_1 \release_{[t_1,t_2]} \varphi_2}((\bm{y},t)) \leq \rho^{\varphi_1  \release_{[t_1,t_2]} \varphi_2}((\bm{y},t))$.}
    \end{itemize}
    
    Since $\rho^{\varphi}((\bm{y},0)) > 0$ is sufficient for satisfaction of $\varphi$, it follows that
    \begin{equation*}
        \tilde{\rho}^\varphi((\bm{y},0)) > 0 \implies \rho^\varphi((\bm{y},0)) > 0 \implies \bm{y} \vDash \varphi,
    \end{equation*}
    and so the Theorem holds. 
    \end{proof}

The advantage of this soundness property is that if a proposed solution is found to have positive smooth robustness, the user can be assured that the trajectory satisfies all constraints. Other robustness approximations either retain this property at the expense of smoothness everywhere (as in the case of AGM robustness) or they sacrifice it for smoothness (as in the case of LSE robustness). Our proposed approximation, on the other hand, is both smooth and sound. 

\begin{figure*}
    \centering
    \begin{subfigure}{0.48\textwidth}
        \centerline{\includegraphics[width=0.6\linewidth]{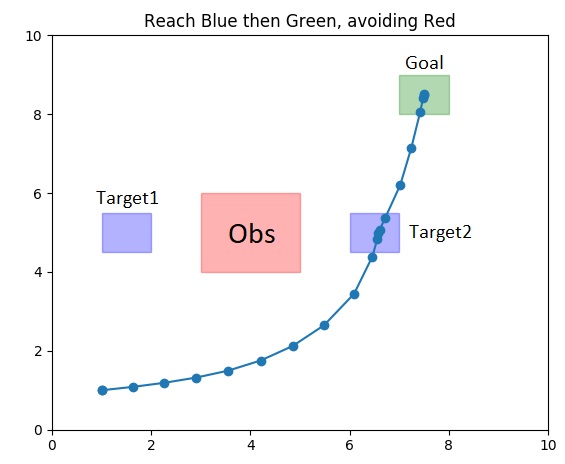}}
        \caption{Locally optimal trajectory with respect to our smooth robustness (Def. \ref{def:smooth_semantics}) for relatively low parameter values ($k_1=k_2=2$). The robot takes several extra steps in the blue target and green goal regions, suggesting greater robustness to external disturbances.{The specification is: $\varphi_{two-target}\always_{[0,10]} ( \lnot Obs ) \land \eventually_{[0,10]} (Target1 \lor Target2) \land \eventually_{[0,10]} (Goal) \land \always_{[0,10]} (-2 \leq u \leq 2)$}.}
        \label{fig:either_or_ours_low_k}
    \end{subfigure}
    \begin{subfigure}{0.48\textwidth}
        \centerline{\includegraphics[width=0.6\linewidth]{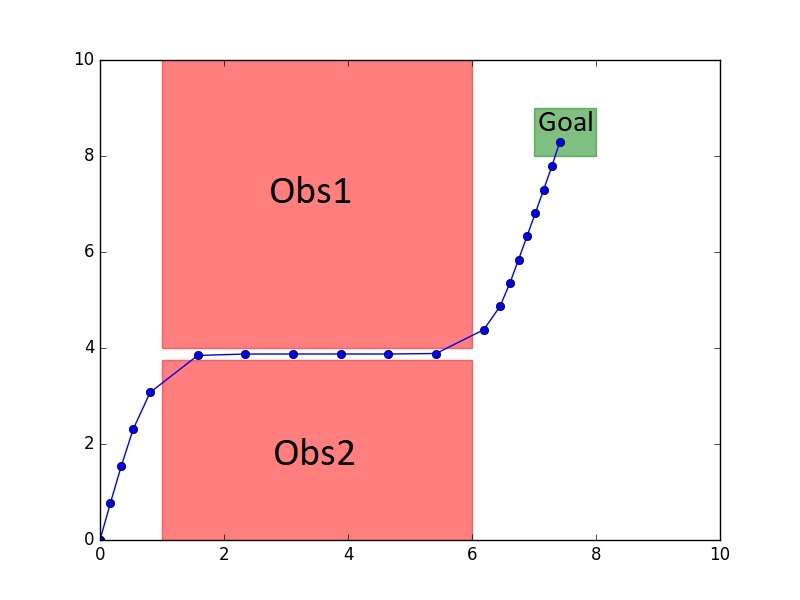}}
        \caption{Locally optimal trajectory with respect to our smooth robustness (Def. \ref{def:smooth_semantics}) for relatively high parameter values ($k_1=k_2=10$). In this situation being too conservative would have obfuscated the solution. The specification is: { $\varphi_{tunnel} = \always_{[0,20]}\lnot (Obs1 \lor Obs2) \land \eventually_{[0,20]} (Goal) \land \\ \always_{[0,10]} (-1 \leq u \leq 1)$}}
        \label{fig:tunnel_ours}
    \end{subfigure}
    \caption{Comparison of our approach with low parameter values (\ref{fig:either_or_ours_low_k}) and high parameter values (\ref{fig:tunnel_ours}). Low values replicate the desired averaging effect of the AGM approach. High values mimic the asymptotic completeness of the LSE approach. The optimizer used the cost function: $J = \tilde{\rho}^{\varphi} + 0.01 u^Tu$}
    \label{fig:our_k_examples}
\end{figure*}

Additionally, our proposed smooth robustness is asymptotically complete, as formalized in the following Theorem:
\begin{theorem}
    For any discrete time signal $\bm{y}$, any finite STL specification $\varphi$, and real-valued scalar $\epsilon > 0$, there exists $\bar{k}_1,\bar{k}_2$ such that $|\rho^\varphi((\bm{y},0)) - \tilde{\rho}^\varphi((\bm{y},0))| \leq \epsilon$ for all $k_1 \geq \bar{k}_1$, $k_2 \geq \bar{k}_2$.
    \label{thmCmplt}
\end{theorem}
\begin{proof}
    Recall from Lemma \ref{lemma:min} that for $k_1 \geq \frac{\log(m)}{\epsilon}$, we have
    \begin{equation}\label{eq:epsilon_bound_min}
        | \min(\mathbf{a}) - \widetilde{\min}(\mathbf{a}) | \leq \epsilon.
    \end{equation}
    Similarly, from Lemma \ref{lemma:max} we have that for $k_2 \geq \frac{1}{a_1-a_2}\log\left(\frac{a_1-a_m}{\epsilon}(m-1)-1\right)$,
    \begin{equation}\label{eq:epsilon_bound_max}
        | \max(\mathbf{a}) - \widetilde{\max}(\mathbf{a}) | \leq \epsilon.
    \end{equation}
    
    With this in mind, it is easy to see that the Theorem holds by induction on the definition of $\tilde{\rho}^\varphi$:
    \begin{itemize}
        \item $|\tilde{\rho}^\pi((\bm{y},t)) - \rho^\pi((\bm{y},t))| = 0$
        \item $|\tilde{\rho}^{\lnot \pi}((\bm{y},t)) - \rho^{\lnot \pi}((\bm{y},t))| = 0$
        \item $|\tilde{\rho}^{\varphi_1 \land \varphi_2}((\bm{y},t)) - \rho^{\varphi_1 \land \varphi_2}((\bm{y},t))| \leq \epsilon$ by (\ref{eq:epsilon_bound_min})
        \item $|\tilde{\rho}^{\varphi_1 \lor \varphi_2}((\bm{y},t)) - \rho^{\varphi_1 \lor \varphi_2}((\bm{y},t))| \leq \epsilon$ by (\ref{eq:epsilon_bound_max})
        \item $|\tilde{\rho}^{\varphi_1 \until_{[t_1,t_2]} \varphi_2}((\bm{y},t)) - \rho^{\varphi_1  \until_{[t_1,t_2]} \varphi_2}((\bm{y},t))| \leq \epsilon$ by (\ref{eq:epsilon_bound_min}) and (\ref{eq:epsilon_bound_max}). 
        \item {$|\tilde{\rho}^{\varphi_1 \release_{[t_1,t_2]} \varphi_2}((\bm{y},t)) - \rho^{\varphi_1  \release_{[t_1,t_2]} \varphi_2}((\bm{y},t))| \leq \epsilon$ by (\ref{eq:epsilon_bound_min}) and (\ref{eq:epsilon_bound_max}).}
    \end{itemize}
\end{proof}

This theorem essentially states that for large $k_i$, our smooth robustness $\tilde{\rho}$ approaches the true robustness $\rho$. This ensures that if $k_1$ and $k_2$ are large enough, then the approximation will not hide any potential solutions. 

\section{Simulation Results}\label{sec:simulation}

To illustrate the effectiveness of our proposed smooth robustness, we use it to plan a path for a simple simulated robot.

{In the first scenario, shown in Figure \ref{fig:either_or_ours_low_k}, a robot must avoid an obstacle (red), visit an intermediate target (blue) and arrive at a goal region (green).} We used a Python implementation of our proposed smooth robustness and found an optimal trajectory using scipy's SQP method \cite{2020SciPy-NMeth}. Gradients were computed symbolically with the autograd \cite{maclaurin2015autograd} package. For small values of $k_1$ and $k_2$, we observe averaging behavior similar to \cite{mehdipour2019arithmetic} (Figure \ref{fig:either_or_ours_low_k}). Notice how the trajectory arrives quickly at each target and stays for as long as possible. This conservativeness improves robustness to disturbances \cite{mehdipour2019arithmetic}. 

For large values of $k_1$ and $k_2$, our smooth robustness closely approximates the true robustness. This asymptotic completeness is a desirable feature of LSE robustness, and enables navigation through narrow passages, as shown in Figure \ref{fig:tunnel_ours}. In this scenario, the robot must pass through a narrow tunnel between red obstacles\footnote{While the trajectories appear to clip the obstacles, this is merely an artifact of the time discritization. In practice, this can be resolved with a smaller sampling period or by inflating the obstacles.} before reaching the green goal region. This problem requires the smooth robustness to be close to the actual robustness: an overly conservative robustness measure will be unable to find a satisfying path through the tunnel. {Interestingly, while the AGM robustness \cite{mehdipour2019arithmetic} is complete, we were unable to find a satisfying run using this approach.}

{To better understand the differences between \cite{pant2017smooth}, \cite{mehdipour2019arithmetic}, and our approach, each was tested for 20 runs or for 100 seconds, whichever generated more runs. The scenario was similar to Figure \ref{fig:either_or_ours_low_k}, except with 11 time steps, no control constraint, and differential drive dynamics. Initial positions were uniformly sampled from $[0,1]$ and the initial orientation was uniformly sampled from $[0,2\pi]$. $k_1=k_2=1$ was used for our approach. Each run used a random control sequence as an initial guess.}
\begin{table}
    {
    \caption{Comparison of smooth robustness approaches}
    \vspace{-1em}
    \begin{center}
        \begin{tabular}{ccccc}
         & $\bar{\rho}$ & $\sigma_{\rho}$ &$\bar{t}$  &  $\sigma_{t}$\\
         EF (ours)  & 0.174 &0.128& 11.876& 0.610 \\
         LSE \cite{pant2017smooth} & -0.120 & 0.336 &11.867 &0.336 \\
         AGM \cite{mehdipour2019arithmetic} & -5.467 &1.240  &0.463 &0.163
        \end{tabular}
    \end{center}
    \label{tab:mulitple_runs}
    }
\end{table}
{Results are shown in Table \ref{tab:mulitple_runs}. Our approach achieved a higher average traditional robustness ($\bar{\rho}$), with lower standard deviation ($\sigma_\rho$). Computation time was similar to LSE robustness, but slower than AGM.\footnote{The optimization times are not directly comparable since the AGM approach was implemented using MATLAB, and optimized with \texttt{fmincon}. The EF and LSE were implemented in Python and optimized with scipy's SQP method \cite{2020SciPy-NMeth}. All gradients were computed numerically.} One notable result is that AGM, while being a sound and complete measure, resulted in significantly lower robustness scores, and had a much higher variance. This is likely due to the existence of many local optima, together with the difficulty of optimizing over a non-smooth objective function.}

\begin{figure}
    \centerline{\includegraphics[width=0.6\linewidth]{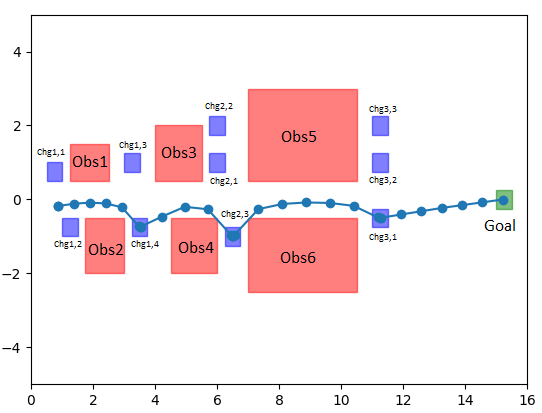}}
    \caption{{A robot, with integrator dynamics, must stop at three charging stations (in blue) on its way to the goal (in green) while avoiding the obstacles (in red). We tested scalability to the number of timesteps and the number of charging stations, results are shown in Figure \ref{Scalability}.}}
    \label{fig:complicated_ex}
\end{figure}

{To evaluate the scalability of our approach, we consider a scenario where a robot must avoid many obstacles and visit numerous charging stations within strict time limits. This scenario is shown in Figure \ref{fig:complicated_ex}. The specification can be formally written as,
\begin{equation}\label{spec:complicated}
    \begin{split}
        \varphi_{goal} &= \eventually_{[0,30]} (Goal)\\
        \varphi_{chg1} &= \eventually_{[1,10]}\big(\always_{[t_1',t_1'+3]}(Chg_{1,1} \lor \dots \lor Chg_{1,4})\big)\\
        \varphi_{chg2} &= \eventually_{[12,17]}\big(\always_{[t_2',t_2'+5]}(Chg_{2,1} \lor Chg_{2,2} \lor Chg_{2,3})\big)\\
        \varphi_{chg3} &= \eventually_{[20,25]}\big(\always_{[t_3',t_3'+3]}(Chg_{3,1} \lor Chg_{3,2} \lor Chg_{3,3})\big)\\
        \varphi_{obs} &=\lnot \always_{[0,30]} \big(Obs_1 \lor \dots \lor Obs_6 \big)\\
        \varphi_{cntrl} &= \always_{[0,30]} (-1 \leq u \leq 1).
    \end{split}
\end{equation}}

{We evaluate the scalability to specification complexity by increasing the number of charging stations ($P$). We evaluate the scalability to specification length by increasing the number of timesteps ($N$). Results from averaging over 100s of runs with random initial states are shown in Figure \ref{Scalability}. Our approach scales roughly linearly in both of these variables. In contrast, MIP-based methods are exponential in both $N$ and $P$. }
\begin{figure}
    \centering
    \begin{subfigure}{0.48\textwidth}
        \centerline{\includegraphics[width=0.8\linewidth]{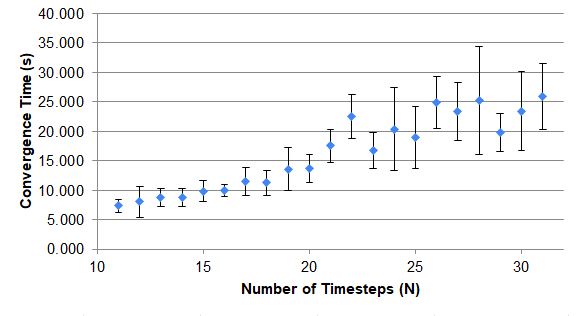}}
        \caption{{Graph relating the number of time steps in a specification to the convergence time with one charging station ($Chg_{1,1}$)}}
        \label{TimeVSN}
    \end{subfigure}
    \begin{subfigure}{0.48\textwidth}
        \centering
        \includegraphics[width=0.8\linewidth]{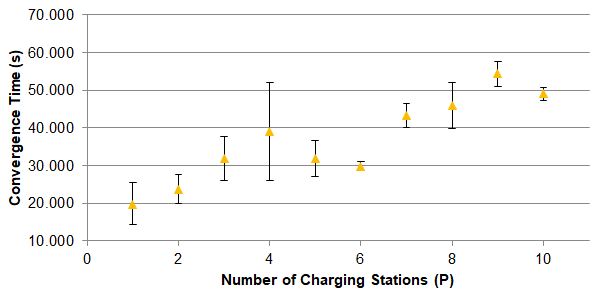}
        \caption{{Graph relating the number of charging stations to the convergence time with a fixed number of time steps, $N=29$.}}
        \label{TimeVSP}
    \end{subfigure}
    \caption{Empirical Scalability Results. As before the optimizer was scipy's SQP method \cite{2020SciPy-NMeth}. This method was initialized with a random initial guess.}
    \label{Scalability}
\end{figure}

\section{Conclusion}\label{sec:conclusion}
We presented a new smooth robustness measure for Signal Temporal Logic. Our proposed robustness measure is sound---a positive value is sufficient for satisfaction---and asymptotically complete---our smooth measure can be arbitrarily close to the true robustness. Our proposed robustness measure combines the benefits of existing smooth approximations of STL robustness, while enabling an explicit tradeoff between conservativeness and completeness. Future work will focus on developing efficient optimization techniques that exploit the structure of our smooth robustness to more effectively find satisfying trajectories. 

\bibliographystyle{./bibliography/IEEEtran}
\bibliography{references}

\end{document}